\documentclass[reqno]{amsart}
\usepackage{amssymb}
\usepackage{amsmath}
\usepackage{url}
\usepackage{graphicx}	 
\makeatletter
\@addtoreset{equation}{section}
\makeatother

\renewcommand\thefigure{\thesection.\@arabic\c@figure}
\renewcommand\thetable{\thesection.\@arabic\c@table}
 
\newtheorem{theorem}{Theorem}[section]
\newtheorem{lemma}[theorem]{Lemma}
\newtheorem{proposition}[theorem]{Proposition}

\newtheorem{remark}[theorem]{Remark}

 \theoremstyle{remark}

\newcommand{\mc}[1]{{\mathcal #1}}

\newcommand{\bb}[1]{{\mathbb #1}}

\newcommand{\<}{\langle}
\renewcommand{\>}{\rangle}

\DeclareMathOperator{\cls}{cls}

\title[Hydrodynamics on a fractal]{Hydrodynamic limit for a zero-range
process in the Sierpinski gasket}
 
\begin{document}

\begin{abstract}

We consider a system of random walks on graph approximations of the Sierpinski gasket, coupled by a zero-range interaction. We prove that the hydrodynamic limit of this system is given by a nonlinear heat equation on the Sierpinski gasekt. 
\end{abstract}


\renewcommand{\subjclassname}{\textup{2000} Mathematics Subject Classification}

\author{Milton Jara}
\address{FYMA, Universit\'e Catholique de Louvain, chemin du Cyclotron 2, B-1348 Louvain-la-Neuve, Belgium} 
\email{milton.jara@uclouvain.be}
\date{}

\maketitle

\section{Introduction}

The Sierpinski gasekt is a fractal in $\bb R^2$ constructed in the following way. Start with an equilateral triangle of side 1. Divide it into 4 equilateral triangles of side 1/2, and retire the central triangle. Repeat this procedure on each of the three remaining triangles. After $n-1$ steps, we are left with $3^n$ small triangles of side $1/2^n$. Since this sequence of triangles is decreasing, and each element in the sequence is compact, there is a non-viod limiting set $K$. We call this set $K$ the {\em Sierpinski gasket}. 

At each step of this construction, consider the boundary of the resulting set as a graph $\Gamma_n$ in $\bb R^2$, with $3(3^n+1)/2$ vertices and $3^{n+1}$ bonds. 
Consider now a system of particles evolving on this graph. The particles are attempting jumps to neighboring sites at rates that depend only on the number of particles sharing the same site. This is the so-called {zero-range} process. Such a system has been extensively studied in the usual lattice $\bb Z^d$ or its periodic version (see \cite{KL} and the references therein). The purpose of this article is to study the collective behavior of this system as the graph get finer and finer, approximating in this way the Sierpinski gasket $K$. More precisely, we are interested in the {\em hydrodynamic limit} of the model, that is, the macroscopic evolution of the density of particles. It turns out that the hydrodynamic limit for this model is given by a nonlinear heat equation of the form $\partial_t u = \Delta \phi(u)$, where $\Delta$ is the Laplacian defined in $K$ and $\phi$ depends on the particular form of the interaction between particles. In a private communication, we have obtained Gaussian fluctuations for the hydrodynamic limit, in the stationary situation.

At first sight, this result does not appear surprising, since the hydrodynamic equation for the zero-range process in the usual lattice $\bb Z^d$ is also given by teh same nonlinear heat equation. However, the Sierpinski gasket $K$ is a fractal, and therefore it does not have a differentiable structure. By this reason, even the definition of the Laplacian $\Delta$ in $K$ is subtle. 
In \cite{BP}, Barlow and Perkins have defined a Brownian motion in $K$ as the scaling limit of the simple random walks defined in $\Gamma_n$. As a by-product, they have defined the Laplacian $\Delta$ in $K$ as the generator of this Brownian motion. A purely analytical definition can be found in the book by Kigami \cite{Kig}. One remarkable fact is that the scaling limit is {\em subdiffusive}, in the sense that the scaling factor, which is equal to $5^n$ in this case, is larger than the square of the mesh of the graph, which is $2^n$. This fact leads to anomalous diffusion properties of the corresponding heat equation, and to the failure of the usual Gaussian estimates on the decay to equilibrium of solutions of the heat equation. 

Although very detailed information about the fundamental solutions of $\partial_t u = \Delta u$ in $K$ has been obtained \cite{Kig2}, up to our knowledge existence and uniqueness of the Cauchy problem for the hydrodynamic equation $\partial_t u= \Delta \phi(u)$ have not been considered in the literature. We provide in this article  the existence and uniqueness results for the Cauchy problem in $K$ that we need to make sense of the  hydrodynamic limit for the zero-range process. Although collateral to our work on the zero-range process, these results could be of independent interest.

From the point of view of interacting particle systems, the zero-range process is an example of a process satisfying the {\em gradient condition}. Roughly speaking, the current of particles is the gradient of another local function (the interaction rate, in this case). Therefore, Fick's law is satisfied at a microscopic level. The classical method to deal with the hydrodynamic limit of gradient systems was introduced in \cite{GPV} and it is based on the so-called {\em one-block} and {\em two-blocks} estimates. Unfortunatedly, the two-blocks estimate does not seem to hold for the Sierpinski gasket. In fact, this two-blocks estimate is based on the {\em moving particle lemma}, which states that a particle can be moved from one site to another with a diffusive entropy cost. The Laplacian $\Delta$ satisfies a Poincar\'e inequality in $K$, and therefore, the spectral gap for the associated particle system is expected to be of the right order. However, in the Sierpinski gasket there are {\em hot spots}, that is, sites of the graph that have to be visited in order to connect two different regions of the graph. In particular, it is not true that the best strategy in order to transport a particle from one site to another is just to follow the shortest path, in contrast with the situation in the integer lattice.

Probably the simplest method to prove hydrodynamic limits of gradient systems is the $H_{-1}$-norm method, due to Chang and Yau \cite{CY} (see \cite{GQ} for a more readable exposition). The advantage of this method is that only requires the one-block estimate; its main drawback is that it works only for diffusive systems without a drift term. The idea of this method is very simple: let $u$, $v$ be two solutions of the hydrodynamic equation. Then, the time derivative of $\<u-v,(-\Delta)^{-1}(u-v)\>$ is equal to $-2\<u-v,\phi(u)-\phi(v)\>$ and in particular, it is decreasing. Therefore, if $u_0=v_0$, then $u_t=v_t$ for any $t>0$. The point is that this relation holds also at a microscopic level. However, at the microscopic level and for non-transitive graphs, there is a correction term involving $\Delta \mc G(x,x)$, where $\mc G$ is the Green function associated to $\Delta$. For the Laplacian in subsets of $\bb R^d$, $\mc G(x,x)$ may not be well defined, but anyway $\Delta G(x,x)$ can be defined, since the singularity of $\mc G(x,x)$ at the diagonal is always of the same magnitude, and cancels, at least in a weak sense, when taking the Laplacian as the limit of averaged differences around $x$. This is not the situation in the Sierpinski gasket $K$. We will see that, despite the fact that the Green function is continuous in $K \times K$, the function $\mc G(x) = \mc G(x,x)$ satisfies $\Delta \mc G(x) =+\infty$, in a sense to be precised later. 
This nonregularity of the Green function poses an extra difficulty to the proof of the hydrodynamic limit. Usually, the one-block estimate allows us to replace a local function of the number of particles by averages over small boxes, {\em when averaged with respect to a continuous test function}. In our case, we need to average it with respect to discrete approximations of $\Delta \mc G(x)$, with we have seen that is not continuous. Therefore, the one-block estimate needs to be proved without averaging with respect to test functions. This has been accomplished only recently \cite{JLS}. The proof however, confines us to dimension $d<2$. Fortunately, the Sierpinski gasket has Hausdorff dimension $d_H=\log 3 /\log 2 <2$, and the so-called {\em local} one-block estimate holds in our case.

This paper is organized as follows. In Section \ref{s1}, we define the Laplacian operator in $K$ and we review some aspects of functional analysis on $K$ that will be needed in the sequel. In particular, we define with some detail the Green function $\mc G(x,y)$ associated to the Dirichlet Laplacian in $K$. Although most of the material of this section has been taken from \cite{Kig}, \cite{Kig2}, we have decided to include most of the proofs for the convenience of readers interested in interacting particle systems, not familiar with analysis on fractals.

In Section \ref{s2} we define what we understand by a weak solution of the hydrodynamic equation $\partial_t u =\Delta \phi(u)$ and we prove existence and uniqueness of such solutions by considering a finite-difference numerical scheme to approximate those solutions.

In Section \ref{s3} we introduce the zero-range process, the $H_{-1}$-norm method and we prove the hydrodynamic limit for the zero-range process, relying on the one-block estimate and suitable properties of the Green function $\mc G(x,y)$. In Section \ref{s4} we prove the one-block estimate and the lemmas needed for the derivation of the hydrodynamic limit.
In the Appendix we study the behavior of $\Delta_n \mc G(x,x)$, where $\Delta_n$ corresponds to the discrete approximation of $\Delta$ defined in $\Gamma_n$. The results presented in the Appendix are not needed for the hydrodynamic limit. Ee have included them here to stress that the example of the Sierpinski gasket is the worst possible case for the derivation of the hydrodynamic limit with the $H_{-1}$-norm method. We believe that the ideas presented here can be adapted to treat general non-homogeneous graphs, like random graphs and trees, percolation clusters and disordered lattices. 

We have organized the paper in such a way that Sections \ref{s2} and \ref{s3} are independent. Therefore, the readers interested in the zero-range process can take for granted the existence and uniqueness results for the hydrodynamic equation and jump directly from Section \ref{s1} to Section \ref{s3}.

\section{The Sierpinski gasket}
\label{s1}
Let $a_0 =(0,0)$, $a_1=(1/2,\sqrt 3/2)$, $a_2=(1,0)$ be the vertices of an
equilateral triangle of unit side in $\bb R^2$. Define $\varphi_i: \bb R^2 \to
\bb R^2$ by taking $\varphi_i(x) = (x+a_i)/2$, $i=0,1,2$. The Sierpinski gasket $K$
is defined as the unique non-empty compact subset $K$ of $\bb R^2$ such
that
\[
K = \bigcup_{i=0,1,2} \varphi_i(K).
\]

A constructive definition of $K$ is the following. Define $V_n \subseteq \bb
R^2$ recursively by taking $V_0=\{a_0,a_1,a_2\}$, $V_{n+1} = \cup_{i} \varphi_i(V_n)$.
Define $V^* = \cup_{n} V_n$. Since $V_n \subseteq V_{n+1}$, it is not hard to
see that $K = \cls(V^*)$, the closure of $V^*$ under the usual topology of $\bb
R^2$. 

Consider $V_n$ as the set of vertices of a non-oriented graph $\Gamma_n =
(E_n,V_n)$, and define inductively the set of bonds of $\Gamma_n$ by taking $E_0 = \{\<a_0a_1\>,\<a_1a_2\>,\<a_2a_0\>\}$ and
\[
E_{n+1} =\{\<\varphi_i(x)\varphi_i(y)\>; \<xy\> \in E_n, i=0,1,2\}.
\]

We say that $\Gamma_n$ is the $n$-th discrete approximation of $K$. For $x, y
\in V_n$, we say that $x \sim_n y$ if $\<xy\> \in E_n$. We simply write $x\sim
y$ when there is no risk of confusion. We say that $x$ and $y$ are {\em
neighbors} in that case.

\subsection{The Laplacian operator in $K$}

Let $u: V^* \to \bb R$ be an arbitrary function. For each $n \geq 0$, we define
\[
\mc E_n(u,u) = (5/3)^n \sum_{x \sim_n y} \big( u(y)-u(x)\big)^2.
\]

\begin{proposition}
For each $n \geq 0$ and each $u: V^* \to \bb R$, 
\[
\mc E_{n+1}(u,u) \geq \mc E_n(u,u).
\]

Moreover, given $n \geq 0$, $u: V^* \to \bb R$, there exists a unique function
$\bar u_n: V^* \to \bb R$ such that
\[
\mc E_n(u,u) = \mc E_{n+p}(\bar u_n, \bar u_n) \text{ for all }  p \geq 0
\]
and $\bar u_n(x)=u(x)$ for every $x \in V_n$.

\label{p2}
\end{proposition}
\begin{proof}
Let us consider the following setup. For ${\bf \alpha} =
(\alpha_0,\alpha_1,\alpha_2)$, ${\bf \beta}= (\beta_0,\beta_1,\beta_2)$, define
the function 
\[
I({\bf \alpha},{\bf \beta}) = \sum_{i \neq j} \Big[ (\alpha_i-\beta_j)^2
+\frac{1}{2} (\beta_i-\beta_j)^2 \Big].
\] 

Let ${\bf \alpha} \in \bb R^3$ be fixed. A simple computation shows that
\[
\inf_{{\bf \beta} \in \bb R^3} I({\bf \alpha},{\bf \beta}) =
5/3\big[(\alpha_0-\alpha_1)^2+(\alpha_1-\alpha_2)^2+(\alpha_2-\alpha_0)^2\big],
\]
and the infimum is attained in a single point ${\bf \beta}$, with $\beta_i
=(2\sigma -\alpha_i)/5$, where $\sigma = \alpha_0+\alpha_1+\alpha_2$. The proof
follows easily from this observation and a chaining argument.
\end{proof}

From the previous result, $\mc E(u,u) = \lim_n \mc E_n(u,u)$ is always well
defined, although maybe infinite. Observe that $\# V_n = 3(3^n+1)/2$. In
particular, $\lim_n \# V_n / 3^n = 3/2$. This motivates the following
definition. For each $n \geq 0$ define the positive measure $\mu_n$ in $K$ by
\[
\mu_n(dx) = \frac{1}{3^n} \sum_{x \in V_n} \delta_x(dx),
\]
where $\delta_x(dx)$ is the Dirac mass at $x$. It is not hard to see that
$\mu_n$ converges in the vague topology to a measure $\mu$ in $K$ that coincides
with the Hausdorff measure in $K$. A simple summation by parts shows that for
any $u: V^* \to \bb R$,
\[
\mc E_n(u,u) = - \int u(x) \Delta_n u(x) \mu_n(dx),
\]
where $\Delta_n$ is the {\em discrete Laplacian} in $V_n$:
\[
\Delta_n u(x) = 5^n \sum_{y \in V_n: y \sim_n x} \big[ u(y)-u(x)\big].
\]

\begin{proposition}
\label{p1}
There exists a universal constant $c>0$ such that for any $u: V^* \to \bb R$
with $\mc E(u,u) < +\infty$ we have
\[
\sup_{\substack{x,y \in V^*\\x \neq y}} \frac{|u(y)-u(x)|}{|y-x|^\alpha} \leq
c\mc E(u,u)^{1/2},
\]
where $\alpha= \log(5/3)/2\log 2$.
\end{proposition}

\begin{proof}
For each $x, y \in V_n$, define
\[
R_n(x,y) = \sup_{u: \mc E_n(u,u) \neq 0} \frac{|u(x) - u(y)|^2}{\mc E_n(u,u)}.
\]

By definition, $|u(x)-u(y)| \leq \sqrt{R_n(x,y)} \mc E_n(u,u)^{1/2}$. Notice
that for any constants $a>0$, $b \in \bb R$ we have $\mc E_n(au+b,au+b)= a^2\mc
E_n(u,u)$. Therefore,
\[
R_n(x,y)^{-1} = \inf_{\substack{u(x)=0\\u(y)=1}} \mc E_n(u,u).
\]

Assume that $x \sim_ny$. Considering the function $u(z) = {\bf 1}(z=y)$, we see
that $R(x,y)^{-1} \leq 4(5/3)^n$. By definition, $R_n(x,y)^{-1} \geq (5/3)^n$.
Therefore, 
\[
|u(x)-u(y)| \leq (3/5)^{n/2} \mc E_n(u,u)^{1/2} \leq 2^{-\alpha n} \mc
E_n(u,u)^{1/2}, 
\]
where $\alpha = \log(5/3)/2\log 2$. Remember that $|x-y| = 2^{-n}$ for $x
\sim_n y$. Therefore, we have proved the inequality for $x \sim_n y$ when some 
$n \geq 0$. Using the triangle inequality, we can extend this relation to $x, y
\in V^*$ arbitrary.
\end{proof}

In particular, the previous proposition tells us that any function $u: V^* \to
\bb R$ with $\mc E(u,u) < +\infty$ is uniformly continuous. Therefore, $u$ can
be continuously extended in a unique way to the set $K$. From now on, we
consider $u$ as defined in $K$, and we will assume that $u:K \to \bb R$ with
$\mc E(u,u)<+\infty$ is continuous. Notice that the points $a_i$, $i=0,1,2$ are
different from the other points in $V^*$. In fact, the points $a_i$ have only
two neighbors in $\Gamma_n$ while other points in $V_n$ have 4 neighbors. It is
natural to define $V_0$ as the {\em boundary} of $K$. Define 
\[
H_1^0(K) = \{u: K \to \bb R; \mc E(u,u) < +\infty, u(a_i)= 0 \; \forall i\}.
\]

For $u \in H_1^0(K)$, define $||u||_1= \mc E(u,u)^{1/2}$. Notice that $\mc
E(u,u)=0$ if and only if $u$ is constant in $K$. Therefore, $||\cdot||_1$ is a
norm. By Proposition \ref{p1}, $H_1^0(K)$ is closed under this norm. It is easy
to see that $(H_1^0(K),||\cdot||_1)$ is a Hilbert space, with inner product
given by the polarization identity
\[
\mc E(u,v) = \big(\mc E(u+v,u+v)-\mc E(u-v,u-v)\big)/4.
\]

Denote by $\mc C_0(K)$ the set of continuous functions $u: K \to \bb R$ with
$u(a_i)=0$ for every $i$. We define $\mc L^2(\mu)$ as the completion of $\mc
C_0(K)$ under the norm
\[
||u||_0= \big\{\int u(x)^2 \mu(dx) \big\}^{1/2}.
\]

\begin{proposition}
The space $H_1^0(K)$ is dense in $\mc L^2(\mu)$.
\end{proposition}
\begin{proof}

Notice that $H_1^0(K) \subseteq \mc C_0(K) \subseteq \mc L^2(K)$.
It is enough to see that $H_0^1(K)$ is dense in $\mc C_0(K)$. Take $u \in \mc
C_0(K)$, and define $\bar u_n$ as in Proposition \ref{p2}. Notice that given
three numbers $\alpha_0, \alpha_1, \alpha_2$, $\beta_i = (2\sigma -\alpha_i)/5$ 
is between the maximum and the minimum of $\alpha_i$. Therefore, the function
$\bar u_n$ satisfies a maximum principle, and for every $x \in V^*$,
\[
|u(x) - \bar u_n(x)| \leq |u(x)-u(y)| + \sup_{\substack{x', y' \in V_n\\x'
\sim_n y'}} |u(x')-u(y')|
\]
for some $y \in V_n$. Therefore, 
\[
\sup_{x \in V^*} |u(x) - \bar u_n(x)| \leq 2 \sup_{\substack{x, y \in V^*\\|x-y|
\leq 2^{-n}}} |u(x)-u(y)|,
\]
which goes to $0$ as $n \to \infty$.
\end{proof}

Denote by $\<u,v\>$ the inner product in $\mc L^2(\mu)$. Define in an analogous
way the spaces $\mc L^2(\mu_n)$ and the inner products $\mc E_n(u,v)$,
$\<u,v\>_n$. Remember the identity
\[
\mc E_n(u,u) = - \<u, \Delta_n u\>_n.
\]

By analogy with the case of the real line, we define the {\em Dirichlet
Laplacian} $\Delta: D(\Delta) \subseteq 
\mc L^2(\mu) \to \mc L^2(\mu)$ as the unbounded operator given by
\begin{itemize}
\item[i)] $D(\Delta) = \{ u \in H_1^0(K); \exists c>0 \text{ with } \mc E(u,v)
\leq c||v||_0 \;\forall \; v \in H_1^0(K)\}$.
\item[ii)] For $u \in D(\Delta)$, $\Delta u =h$ if and only if $\mc E(u,v) =
-\<h,v\>$ for every $v \in H_0^1(K)$.
\end{itemize}

Notice that this definition of the Laplacian in $K$ is not constructive. At
this point, even to find a single example of a function $u \in D(\Delta)$ is
hard to achieve. Moreover, for a generic $u \in \mc C_0(K)$, the approximations
$\bar u_n$ of Proposition \ref{p2} are not in $D(\Delta)$. For this reason we
will extend the definition of the Laplacian in the following way. For $u \in \mc
L^2(\mu)$, define the dual norm
\[
||u||^2_{-1} = \sup_{v \in H_1^0(K)} \big\{2 \<u,v\> - \mc E(v,v)\big\}.
\]

By Proposition \ref{p2}, for any $u \in H_1^0(K)$ we have $||u||_\infty \leq
c||u||_1$. Therefore, the {\em Friedrich's inequality} $||u||_0 \leq c||u||_1$
holds, and in particular $||u||_{-1} < +\infty$ for every $u \in \mc L^2(K)$. We
denote by $H_{-1}(K)$ the closure of $\mc L^2(K)$ under this norm. We extend the
definition of $\Delta$ to the operator (still denoted by) $\Delta: H_1^0(K) \to
H_{-1}(K)$ such that
\[
\< \Delta u, v\> = \mc E(u,v) \text{ for all } u,v \in H_1^0.
\]

Notice that $\Delta$ is now well defined as an element of $H_{-1}(K)$ for every
$u \in H_1^0$, and $\Delta$ is an isometry from $H_1^0(K)$ to $H_{-1}(K)$. Since
$\mc L^2(\mu)$ is dense in $H_{-1}$, we conclude that $D(\Delta)$ is dense in
$\mc L^2(\mu)$. For any function $u \in \mc C_0(K)$, we have
\[
\Delta \bar u_n = \sum_{x \in V_n} u(x) \delta_x,
\]
and therefore we are plenty of examples of functions $u$ for which $\Delta u$ can be evaluated (in $H_{-1}$, of course). 
Since $H_1^0 \subseteq \mc C_0(K)$, the set $\mc M^0(K)$ of Radon measures in $K
\setminus V_0$ is contained in $H_{-1}(K)$, and the previous expression makes
sense. We say that $\bar u_n$ is the harmonic continuation of $u|_{V_n}$. Notice
that $\bar u_n$ converges to $u$ in $H_1^0(K)$, and therefore $\Delta \bar u_n$
converges to $\Delta u$ in $H_{-1}(K)$.

\subsection{The {\em carr\'e du champ}}

The set $K$ does not admit a differentiable structure, since it is clear that no neighborhood of $K$ can be put in bijection with an open set of $\bb R^d$ in a differentiable way. Therefore, the notion of a gradient $\nabla u$ for functions $u:K \to \bb R$ seems hopeless. Moreover, since the dimension of $K$ is not an integer, it is not clear how many components should $\nabla u$ have. What is remarkable, is that the so-called {\em carr\'e du champ} $|\nabla u|^2$ can be defined in a very simple way. We say that a set $\mc T \subseteq K$ is a {\em triangle} if $\mc T = K \cup \Delta(x_0,x_1,x_2)$ for some triangle $\Delta(x_0,x_1,x_2)$ with vertices mutually adjacents in $V_n$ for some $n \geq 0$ (that is, $x_i \sim_n x_j$ for $i \neq j$). In an equivalent way, $\mc T$ is a triangle if $\mc T$ is of the form $\varphi_{i_n} \circ \cdots \circ \varphi_{i_1}(K)$ for some sequence $\{i_1,\dots,i_n\}$ in $\{0,1,2\}$. Fix a function $u \in H_{1}(K)$. For each triangle $\mc T$, we define
\[
\int_{\mc T}  d\mu_{[u,u]}  = \lim_{n \to \infty} (5/3)^n\sum_{\substack{x, y \in \mc T\\x \sim_n y}} \big(u(y)-u(x)\big)^2.
\]

By the proof of Proposition \ref{p2}, this sum is increasing and therefore the limit always exists. Moreover, the limit is always finite, since it is bounded by $\mc E(u,u)$. The set of triangles generates the Borel topology in $K$. It is also not hard to check the continuity at vacuum of the set-valued function $\mu_{[u,u]}$. Therefore, we conclude that $\mu_{[u,u]}$ can be extended to a positive, finite measure in $K$. For two given functions $u$, $v$ in $H_1(K)$, we define the measure $\mu_{[u,v]}$ by polarization:
\[
\mu_{[u,v]}= \frac{1}{4}\big(\mu_{[u+v,u+v]} - \mu_{[u-v,u-v]}\big).
\]

It has been shown \cite{Kus} that the measures $\mu_{[u,u]}$ are {\em singular} with respect to the Hausdorff measure $\mu$ for any $u \in H_1(K)$. However, it has been shown that these measures are not mutually singular in the following sense: there exists a measure $\bar \mu$ in $K$ such that for any pair of functions $u$, $v$ in $H_1(K)$ we have
\[
\mu_{[u,v]} = \Gamma(u,v) \bar \mu
\]
for some function $\Gamma(u,v)$ in $\mc L^1(\bar \mu)$. In fact, the measure $\bar \mu$ can be chosen as equal to $\mu_{[h_1,h_1]}+\mu_{[h_2,h_2]}$ for suitable harmonic functions $h_1$, $h_2$. Therefore, we can define $\nabla u \cdot \nabla v = \Gamma(u,v)$ to get the identity
\[
\mc E(u,v) = \int \nabla u \cdot \nabla v d \bar \mu.
\]

\subsection{Harmonic functions and the integration by parts}

Usually, a function $h: K \to \bb R$ is said to be {\em harmonic} if $\Delta h=0$. Notice
that the only function $h$ in $H_1^0(K)$ for which $\Delta h =0$ is $h=0$ and we only have defined $\Delta h$ for functions in $H_1^0(K)$. Therefore,
we need a definition of what we mean by an harmonic function. A function $h: V^* \to
\bb R$ is said to be {\em harmonic} if $\Delta_n h(x)=0$ for every $x \in V_n
\setminus V_0$ and every $n\geq 1$. In that case, $\mc E(h,h)= \mc E_0(h,h)$ and
by Proposition \ref{p2}, $h$ can be uniquely extended to a continuous function
$h: K \to \bb R$. Notice that $h$ is entirely determined by its values at the
boundary $V_0$.

We extend the definition of $\Delta$ as follows. For a
continuous function $u$ not necessarily in $H_1^0(K)$, we say that $\Delta u =
v$ if there exists an harmonic function $h$ such that $u-h \in H_1^0(K)$ and
$\Delta(u-h)=v$. In this case we say that $u \in H_1(K)$. 

For a function $u: V_n \to \bb R$,
we define the discrete Dirichlet Laplacian by $\Delta^D_n u(x) =
\Delta_n u(x)$ if $x \notin V_0$ and $\Delta^D_nu(x)=0$ if $u \in V_0$. Define
the normal derivatives $\partial^i_n u$ by
\[
\partial^i_n u = (5/3)^n \sum_{\substack{y \in V_n\\y \sim_n a_i}} u(y) -u(a_i).
\]

We have the following (discrete) integration by parts formula:
\[
\<u, \Delta_n^D v\>_n = \<v, \Delta_n^D u\>_n + \sum_{i=0,1,2} \Big( u(a_i)
\partial^i_n v - v(a_i) \partial^i_n u\Big). 
\]

In order to obtain an analogous formula for the Dirichlet Laplacian, for $u \in
H_1^0(K)$ we define
\[
\partial^i u = \< \Delta u, h^i\>,
\]
where $h^i$ is the harmonic function with $h^i(a_j) = \delta_{ij}$. Since
$\Delta$ is symmetric, $\<u,\Delta v\> = \< v, \Delta u\>$ for any pair of
functions $u, v \in H_1^0(K)$. It is straightforward to check the identity
\[
\<u, \Delta v\> -\<v, \Delta u\> = \sum_{i=0,1,2} \big\{ u(a_i) \partial^i v -
v(a_i) \partial^i u\big\}
\]
for any two functions $u,v \in H_1(K)$.

\subsection{The Green function in $K$} 
Friedrich's inequality tells us that the Dirichlet Laplacian has a positive
spectral gap in $\mc L^2(\mu)$. In particular, for every $u \in \mc L^2(\mu)$,
the equation
\begin{equation}
\label{ec6}
\begin{cases}
-\Delta w = u\\
w \in H_1^0(K)\\
\end{cases}
\end{equation}
has a unique solution. Since the inclusion $H_1^0(K) \subseteq \mc L^2(\mu)$ is compact, we conclude that the operator $(-\Delta)^{-1}$ is compact. We put the $-$ sign to emphasize that $\Delta$ is non-positive. In particular, there are an orthonormal basis $\{v_i\}_i$ of $\mc L^2(\mu)$ and a non-decreasing sequences $\{\lambda_i\}_i$ of positive numbers such that $-\Delta v_i = \lambda_i v_i$ for any $i$.
The function $w$ can be written in terms of the
orthonormal basis $\{v_i\}_i$:
\[
w = \sum_{i \geq 1} \lambda_i^{-1}\<u,v_i\>v_i.
\]

Formally, we can obtain $w(x)$ by an integral formula:
\[
w(x) = \int \mc G(x,y) u(y) \mu(dy), \text{ where }
\mc G(x,y)= \sum_{i \geq 1} \lambda_i^{-1} v_i(x)v_i(y)
\]
is the Green function associated to the Dirichlet Laplacian $\Delta$ in $K$. A
simple computation shows that the sum defining $\mc G(x,y)$ is convergent in
$\mc L^2(\mu \otimes \mu)$ if $\sum_{i \geq 1} \lambda_i^{-2} < +\infty$.
We will give a constructive definition of $\mc G(x,y)$ that will allow us to
prove finer properties of $\mc G(x,y)$. 
Take the discrete Laplacian $\Delta_n$ in $V_n$ and 
define the Green function $\mc G_n(x,y)$ at $x,y \in V_n$ as the solution of
\[
\Delta_n \mc G_n(x,y) =
\begin{cases}
 3^n \delta(x,y) , &x \in V_n \setminus V_0\\
 0, &x \in V_0,
\end{cases}
\]
where $\delta(x,y)=1$ if $x=y$ and $\delta(x,y)=0$ otherwise.
Here the operator $\Delta_n$ acts on the first variable $x$. Notice that $\mc
G_n(x,y)$ is non-negative and $\mc G_n(x,y) \leq \mc G_n(x,x)$ for any $x, y \in V_n$.
We extend the definition of $\mc G_n(x,y)$ to $K$ by taking the harmonic
continuation of $\mc G_n$ given by Proposition \ref{p1}. A key observation is
that for $y \in V_n$, $\mc G_{n+1}(x,y)= \mc G_n(x,y)$ for any $x \in V_{n+1}$.
In fact, for $x \in V_{n+1}$ with $x\neq y$, $\Delta_{n+1} \mc G_n(x,y)=0$.
For $x = y$, a simple computation shows that $\mc G_n$ scales correctly, and
therefore $\Delta_{n+1} \mc G_n(y,y) = 3^{n+1}$.
The following propositions shows that it is sufficient to compute $\mc G_1(x,y)$
to obtain $\mc G_n(x,y)$ for every $n \geq 1$, $x,y \in V^*$:

\begin{proposition}
For any $i=0,1,2$,
\[
\mc G_{n+1}(\varphi_i(x),\varphi_i(y)) = \frac{3}{5} \mc G_n(x,y) +
\sum_{j=0,1,2} h^j(y) \mc G_{1}(\varphi_i(x),\varphi_i(a_j)).
\]
\label{p3}
\end{proposition}

The proof is simple; we refer to Section \ref{s5} for the argument. As a consequence of this relation and the previous discussion, $\mc G_n(x,y)$ does
not really depend on $n$. Therefore, we define for $x,y \in V^*$, $\mc G(x,y)=
\mc G_n(x,y)$, where $n$ is such that $x,y \in V_n$. For fixed $y$, we see that
$\mc E(\mc G(\cdot,y),\mc G(\cdot,y)) = \mc G(y,y).$ In particular, $\mc
G(\cdot,y)$ is uniformly continuous and can be uniquely extended to $K$. In this
way we can not define $\mc G(x,x)$ for $x \notin V^*$. We would like to prove
that in fact $\mc G(x,y)$ is uniformly continuous in $V^* \times V^*$. This is
an immediate consequence of the following proposition:

\begin{proposition}
There exists a constant $c>0$ such that $\mc G(x,x) \leq c$ for every $x \in
V^*$. 
\end{proposition}
\begin{proof}
Notice that $\sum_i h^i(x) =1$ for every $x \in K$. In fact, $\sum_i h_i$
corresponds to the harmonic function with $h(a_i)=1$ for $i=0,1,2$, which is identically constant. Taking $x=y$
in Proposition \ref{p3}, we see that
\[
\mc G(\varphi_i(x),\varphi_i(x)) = \frac{3}{5} \mc G(x,x) + \sum_{j=0,1,2} h^j(y) \mc
G(\varphi_i(x), a_i).
\]

Since every $y \in V_{n+1} \setminus V_n$ is equal to $\varphi_i(x)$ for some
$i$ and some $x \in V_n \setminus V_{n-1}$, we see that
\[
\sup_{y \in V_{n+1}\setminus V_n} \mc G(y,y) \leq 3/5\sup_{x \in V_n \setminus
V_{n-1}} \mc G(x,x)
+ \sup_{j=0,1,2} \sup_{x \in K} \mc G(x,\varphi_i(a_j)).
\]

Setting $\beta_n = \sup_{x \in V_n\setminus V_{n-1}} \mc G(x,x)$, we see that
$\beta_{n+1} \leq 3/5 \beta_n +c$ for some constant $c$ independent of $n$. A
simple computation shows that $\beta_n$ is bounded in $n$.
\end{proof}

Remember that for any function $u \in H_1^0(K)$, $\Delta \bar u_n$ converges to
$\Delta u$ in $H_{-1}(K)$. Therefore, $\Delta \mc G(x,y) = \delta_y(x)$ and
$w(x) = \int \mc G(x,y) u(y) \mu(dy)$ is the solution of equation (\ref{ec6}), at
least for functions $u \in \mc C_0(K)$. By the continuity of $\mc G(x,y)$, we
conclude that $\int \mc G(x,y) u(y) \mu(dy)$ solves (\ref{ec6}) for $u \in \mc
L^2(\mu)$ as well.

\section{The nonlinear heat equation in $K$}
\label{s2}

Let $\phi: \bb R_+ \to \bb R_+$ be a smooth function. We will assume that there exists $\epsilon_0 >0$ such that $\epsilon_0 \leq \phi'(u) \leq \epsilon_0^{-1}$ for any $u \in \bb R_+$. Fix some $T>0$. We want to study the Cauchy problem
\begin{equation}
\label{echid}
\begin{cases}
\partial_t u &= \Delta \phi(u)\\
u(t,a_i) &= \alpha_i, i=0,1,2\\
u(0,\cdot) &= u_0(\cdot).\\
\end{cases}
\end{equation}

More precisely, we want to obtain criteria for existence and uniqueness of solutions for this equation. Now we define what do we understand by a weak solution of (\ref{echid}). We say that $u:[0,T] \times K$ is a weak solution of (\ref{echid}) if:
\begin{itemize}
\item[i)] For almost every $t \in [0,T]$, $u(t,\cdot) \in H_1(K)$ and
\[
\int_0^T ||u(t,\cdot)||_1^2 dt <+\infty.
\]
\item[ii)] For any function $G:[0,T] \to H_1^0(K)$, pointwise differentiable in $t$ and strongly differentiable in $H_{-1}(K)$ as a function of $[0,T]$,
\[
\<u_T,G_T\> -\<u_0,G_0\> - \int_0^T\big\{ \<u_t,\partial_t G_t\> + \< \phi(u_t),\Delta G_t\>\big\} dt = \sum_{i=0,1,2} \alpha_i \partial^i G_t.
\]
\end{itemize}

We will start by considering a finite-difference scheme that approximates equation (\ref{echid}).

\subsection{A discrete nonlinear equation}
Take a function $u_0^n: V_n \to \bb [0,\infty)$ such that $u_0^n(a_i) = \alpha_i$. Let us define $u^n(t,x): [0,\infty) \times V_n \to [0,\infty)$ as the solution of the following system of ordinary differential equations:
\[
\left\{
\begin{aligned}
\frac{d}{dt} u^n(t,x) &= \Delta_n \phi(u^n(t,x)) \text{ for } x \in V_n^0\\
u^n(t,a_i) &= \alpha_i \text{ for } i=0,1,2\\
u^n(0,x) &= u_0^n(x).
\end{aligned}
\right.
\]

By the maximum principle and Peano's theorem, $u^n(t,x)$ is well defined for any $t>0$. We will prove existence of solutions for equation (\ref{echid}) in a proper sense by taking limits of these approximated solutions $u^n(t,x)$. To avoid an overcharged notation, we will take $\alpha_i=0$. Our arguments work for $\alpha_i$ arbitrary as well: just take into account the boundary terms when performing integrations by parts. Let us define the discrete norms
\[
||u||_{0,n}^2 = \int u(x)^2 \mu_n(dx) = \frac{1}{3^n} \sum_{x \in V_n} u(x)^2,
\] 
\[
||u||_{1,n}^2 = \mc E_n(u,u) = \frac{5^n}{3^n} \sum_{x \sim_n y} (u(y)-u(x))^2.
\]

Let us denote the function $u^n(t,\cdot)$ by $u^n_t$. It is not hard to see that $||u_t^n||_{0,n}$ is decreasing. In fact,
\begin{align*}
\frac{d}{dt}||u_t^n||_{0,n}^2 &= 2\<u_t^n, \Delta_n \phi(u_t^n) \>_n\\
	&=-2\mc E_n(u_t^n,\phi(u_t^n)) \leq -2\epsilon_0 ||u_t^n||_{1,n}^2.
\end{align*}

Integrating this inequality between $t=0$ and $t=T$, we see that
\begin{equation}
\label{ec4}
||u_t^n||_{0,n}^2 + 2\epsilon_0 \int_0^T ||u_t^n||_{1,n}^2dt \leq ||u_0||_{0,n}^2. 
\end{equation}

Let us fix some reference time $T>0$. For a function $u: [0,T] \times K \to \bb R$, define
\[
|||u|||_{1}^2 = \int_0^T ||u||_1^2,
\]
and denote by $H_{1,T}^0(K)$ the Hilbert space obtained as the closure of $\mc C([0,T],H_{1}^0(K))$ with respect to this norm , where $\mc C([0,T],H_{1}^0(K))$ denotes the space of continuous paths in $H_{1}^0(K)$. Fix a continuous function $u_0: K \to [0,\infty)$ with $u(a_i) = 0$ for $i=0,1,2$. Consider $\bar u_t^n$, the harmonic continuation of $u_t^n$ into $K$. By (\ref{ec4}), we have
\[
\sup_n |||\bar u_t^n|||_1^2 \leq \frac{||u_0||_\infty}{2\epsilon_0}.
\]

In particular, there is a subsequence $n'$ such that $u^{n'}_t$ converges to some function $u_t \in H_{1,T}(K)$, weakly with respect to the topology of $H_{1,T}(K)$. 

\begin{theorem}
The function $u_t \in H_{1,T}(K)$ is a weak solution of (\ref{echid}).
\end{theorem}

\begin{proof}
Taking a second subsequence if necessary, we can assume that $\bar \phi(u_t^n)$ converges to some function $w_t$ as well, where $\bar \phi(u_t^n)$ is the harmonic continuation of $\phi(u_t^n)$. At this point, we need to justify the identity $w_t=\phi(u_t)$. For the usual Laplacian, defined on a bounded, open set $U \subseteq \bb R^d$, the argument is the following. If $u \in H_1^0(U)$, that is, if $\< u, -\Delta u\> < +\infty$, then there exists a {\em unique} function $\nabla u: U \to \bb R^d$ such that $\<u, \Delta G\> = \<\nabla u, \nabla G\>$ for any smooth function $G$. Since $\phi(u)$ is a smooth function of $u$, $\phi(u)$ also belongs to $H_1^0(U)$, and moreover $\nabla \phi(u) = \phi'(u) \nabla u$. The uniqueness of the {\em weak gradient} $\nabla u$ would allow us to conclude that $w_t = \phi(u_t)$. 

Recall the definition of the {\em carr\'e du champ} $\nabla u$. A simple Taylor expansion shows that 
\[
\int_{\mc T} d\mu_{[\phi(u),\phi(u)]} = \int_{\mc T} \phi'(u) d\mu_{[u,u]},
\]
as expected. Therefore, we can appeal to the uniqueness of the representation $\mu_{[u,u]}=\Gamma(u,u) \bar \mu$ to conclude that $w_t = \phi(u_t)$.

Take a function $G_t \in H_{1,T}(K)$. Assume that $G_t$ is of class $\mc C^1$ in time. 
By hypothesis,
\[
\int_0^T \int_K \bar \phi(u_t^n(x)) \Delta G_t(x) \mu(dx) ds \xrightarrow{n \to \infty} \int_0^T \int_K \phi(u_t(x)) \Delta G_t(x) \mu(dx) dt.
\]

Performing an integration by parts, we see that the left-hand side of the previous expression is equal to
\begin{align*}
\int_0^T \int_K G_t(x) \Delta \bar \phi(u_t^n(x)) \mu(dx) ds 
	&= \int_0^T \<\frac{d}{dt} u_t^n,G_t\>_n dt\\
	&= \<u_T^n,G_T\>_n- \<u_0,G_0\>_n- \int_0^T \<u_t^n, \partial_t G_t\>_n dt.
\end{align*}

Remember that $||u_T^n||_{0,n}$ is also bounded by $||u_0||_\infty$. In particular, choosing a further subsequence if necessary, we can assume that $u_T^n$ converges weakly to $u_T$ in $\mc L^2(\mu)$. Therefore, $\<u_T^n,G_T\>_n$ converges to $\<u_t,G_T\>$. By Friedrich's inequality, weak convergence in $H_{1,T}(K)$ is stronger than weak convergence in $\mc L^2(\mu(dx)\times dt)$. Therefore, we can pass to the limit in each of the terms on the right-hand side of the previous expression. We have therefore proved that
\begin{equation}
\label{ec5}
\<u_T,G_T\>-\<u_0,G_0\> - \int_0^T \big\{ \<u_t, \partial_t G_t\> + \<\phi(u_t), \Delta G_t\> \big\} dt =0
\end{equation}
for any function $G_t$ smooth enough, which proves the theorem.
\end{proof}

\begin{theorem}
The equation (\ref{echid}) has at most one weak solution. 
\end{theorem}
\begin{proof}
The heuristic argument is very simple. The idea is to formalize the following formal computation. Take two weak solutions $u_t$, $v_t$ of (\ref{echid}). Then,
\begin{equation}
\label{ec7}
\begin{split}
\frac{d}{dt} ||u_t-v_t||_{-1}^2 
	&= 2\<(-\Delta)^{-1}(u_t-v_t), \Delta(\phi(u_t)-\phi(v_t))\>\\
	&= -2\<u_t-v_t, \phi(u_t)-\phi(v_t)\> \leq 0.
\end{split}
\end{equation}

Therefore, if $u_0=v_0$, we conclude that $u_t=v_t$ for any $t>0$. Of course this heuristic computation needs to be justified. By (\ref{ec5}), we have
\[
\<u_T-v_T,G_T\>= \int_0^T \big\{ \<u_t-v_t, \partial_t G_t\> + \<\phi(u_t)-\phi(v_t), \Delta G_t\> \big\} dt.
\]

Let us take $G_t = (-\Delta)^{-1}(u_t -v_t)$ in the previous expression. That is, define
\[
G_t(x) =\int_K \mc G(x,y)\big(u_t(y) -v_t(x)\big)\mu(dy).
\]

Putting this into the previous formula, we obtain immediately
\[
||u_T-v_T||_{-1}^2 = -2\int_0^T \<u_t-v_t,\phi(u_t)-\phi(v_t)\>dt,
\]
which is just the integral version of (\ref{ec7}). But we still need to justify that $G_t$ can be taken as a test function. Since $u_t$ and $v_t$ are in $H_1(K)$ for almost every $t \in [0,T]$, for fixed $t$ the function $G_t$ is regular enough. The problem is that $G_t$ maybe is not differentiable in $t$. This problem solves easily by taking an approximation of the identity $\gamma_\delta(t)$ with support in $[0,\delta]$ and defining 
\[
G_t^\delta = \int_0^\delta \gamma_\delta(s) G_{t+s} ds.
\]

Now $G_t^\delta$ is differentiable, so it is an admissible test function. Taking $\delta \to 0$ we obtain the desired result.
\end{proof}

\section{Hydrodynamic limit for the zero-range process}
\label{s3}
\subsection{The zero-range process}
Let $g: \bb N_0 = \{0,1,\cdots\} \to [0,\infty)$ be a function with $g(0)=0$. The zero-range process in $V_n$ with interaction rate $g(\cdot)$ is defined as the continuous-time Markov chain $\xi_t$ in $\overline{\Omega}_n = \bb N_0^{V_n}$ and generated by the operator 
\[
L_{zr}^b = \sum_{x\in V_n} \sum_{y \sim_n x} g\big(\xi(x)\big)\big[f(\xi^{x,y})-f(\xi)\big], 
\]
where $\xi$ is a generic element of $\overline \Omega_n$, $f:\overline \Omega_n \to \bb R$ and $\xi^{x,y}$ is given by
\[
\xi^{x,y}(z) =
\begin{cases}
\xi(x)-1, &z=x\\
\xi(y)+1, &z=y\\
\xi(z), &z \neq x,y.\\
\end{cases}
\]

Notice that the number of particles in this process is preserved by the dynamics. Therefore, for any fixed initial configuration, the state space is finite, and the previous process is well defined. This process has a family of invariant measures which we describe as follows. Define $g(n)!= g(1)\cdots g(n)$, $g(0)!=0$. Assume that 
\[
\phi^*=\Big\{\limsup_{n \to \infty} \sqrt[n]{g(n)!}\Big\}^{-1}
\]
is non-zero. This is the fact if, for example, $\inf_{n \geq n_0} g(n) >0$ for some $n_0$. For any $\phi < \phi^*$, define the uniform product measure $\bar \nu_{\phi}$ in $\overline \Omega_n$ by
\[
\bar \nu_{\phi}\big\{ \xi; \xi(x) =k\big\} = \frac{1}{Z(\phi)} \frac{\phi^k}{g(k)!},
\]
where $Z(\phi)$ is the normalization constant. Notice that due to the fact that $\phi < \phi^*$, the normalization constant $Z(\phi)$ is finite. It is not hard to see that the measure $\bar \nu_\phi$ is invariant under the evolution of $\xi_t$. Observe that $\phi = \int g(\xi(x)) \bar \nu_{\phi}(d\xi)$.

Define the number of particles per site by $\rho(\phi) = \int \xi(x) \bar \nu_{\phi}(d\xi)$. Notice that the application $\phi \mapsto \rho(\phi)$ is strictly increasing, with $\rho(0)=0$. Therefore, $\phi \mapsto \rho(\phi)$ is a bijection between $[0,\phi^*)$ and $[0,\rho^*)$, where
\[
\rho^* = \lim_{\phi \uparrow \phi^*} \rho(\phi).
\]

Denote by $\rho \mapsto \phi(\rho)$ the inverse mapping of $\rho(\phi)$. 
Since the number of particles per site is a more natural quantity than $\phi$, we define $\nu_\rho = \bar \nu_{\phi(\rho)}$ for $\rho \in [0,\rho^*)$. 

Now we will introduce a Dirichlet-type boundary condition into this process. Fix some numbers $\alpha_i \in [0,\rho^*)$, $i=0,1,2$. Define the boundary operators by
\[
L_{zr}^i f(\xi) = \sum_{y \sim_n a_i} \Big\{ \phi(\alpha_i)\big[f(\xi+\delta_y) -f(\xi)\big] + g\big(\xi(y)\big) \big[f(\xi-\delta_y)-f(\xi)\big]\Big\},
\]
where the configurations $\xi \pm \delta_y$ are given by
\[
\big(\xi \pm \delta_y\big)(z) = 
\begin{cases}
\xi(z) \pm 1, &z=y\\
\xi(z), &z \neq y.\\
\end{cases}
\]

The zero-range process in $V_n$ with boundary conditions $\{\alpha_i\}_i$ is then defined as the continuous-time Markov process $\xi_t$ in $\Omega_n = \bb N_0^{V_n^0}$ and generated by the operator
\[
L_{zr} f(\xi) = \sum_{x \in V_n^0} \sum_{\substack{y \in V_n^0\\y \sim_n x}} g\big(\xi(x)\big) \big[f(\xi^{x,y})-f(\xi)\big] + \sum_{i=0,1,2} L_{zr}^i f(\xi).
\]

The difference between this process and the zero-range process defined previously is easy to understand. Inside $V_n$ (that is, in $V_n^0$), the dynamics is the same. Particles are coming from the boundary sites $a_i$ with intensity $\phi(\alpha_i)$, which corresponds to have a density of particles $\alpha_i$ at $a_i$. Particles are also annihilated when they jump into the boundary sites $a_i$, in order to keep the density of particles at $a_i$ fixed.

Notice now that the number of particles is not longer fixed, since particles are coming in and out at the boundary sites. In order to have a process with amenable properties, we will impose some technical conditions on the interaction rate $g(\cdot)$ (see \cite{KL}). We say that $g(\cdot)$ satisfies {\bf (SG)} condition if
\begin{itemize}
\item[i)] $\sup_k|g(k+1)-g(k)| <+\infty$,
\item[ii)] There exist $k_0>0$ and $a_0>0$ such that $g(k+l)-g(k) \geq a_2$ for any $l>k_0$.
\end{itemize}

Notice that in this case $\rho^* =+\infty$. This condition guarantees the existence of exponential moments for the occupation variables $\xi(x)$ under the invariant measures $\nu_\rho$. We say that $g(\cdot)$ satisfies {\bf (C)} condition if $g(k+1) \geq g(k)$ for any $k$. In this case, there exists a constant $\theta_0>0$ such that $\int\exp\{\theta_0 \xi(x)\}d\nu_\rho <+\infty$ for any $\rho \leq \rho^*$. We will assume throughout this article that $g(\cdot)$ satisfies {\bf (SG)} or {\bf (C)}. 

Condition {\bf (SG)} guarantees the existence of a {\em uniform spectral gap}, which states that the magnitude of the first non-null eigenvalue of $L_{zr}^b$ with respect to $\nu_\rho$ is bounded below by a constant that does not depend on the density $\rho$. We expect this constant to be of order $5^{-n}$. The non-validity of the moving particle lemma, discussed in the Introduction, prevent us to obtain such a bound. Notice, however, that a simple computation shows that there exists a constant $c_0$, independent of $n$ and $\rho$, such that the spectral with respect to $\nu_\rho$ gap is bounded below by $c_0 6^{-n}$.

Condition {\bf (C)} implies that the zero-range process is {\em attractive}, which means that, given two initial configurations $\xi$, $\xi'$ with $\xi(x) \leq \xi'(x)$ for any $x$, there exists a joint process $(\xi_t,\xi_t')$ such that $\xi_t$ is a zero-range process starting from $\xi$, $\xi_t'$ is a zero-range process starting from $\xi'$ and $\xi_t(x) \leq \xi_t'(x)$ for any $x \in V_n$ and any $t >0$. This property allows us to obtain moment bounds for $\xi_t(x)$ in terms of the invariant measures $\nu_\rho$.

Now a simple path argument shows that there is exactly one invariant measure for the evolution of $\xi_t$. Remarkably, this invariant measure is still of product form.
Consider the harmonic function $h$ with $h(a_i) = \alpha_i$. It is not hard to see that the non-uniform product measure $\nu_h$ defined by
\[
\nu_h \big\{\xi; \xi(x) = k\big\} = \frac{1}{Z(h(x))} \frac{h(x)^k}{g(k)!}
\]
is invariant and ergodic for the evolution of $\xi_t$.

\subsection{The $H_{-1}$-norm method: heuristics} 

Probably the simplest method to prove hydrodynamic limits for particle systems of gradient type is the $H_{-1}$-norm method introduced by Chang and Yau \cite{CY} (see \cite{GQ} for a more comprehensible approach). The main drawback of this method is that only works for strictly diffusive systems. The other alternatives are the so-called {\em entropy method} \cite{GPV} and {\em relative entropy method} \cite{Y}. As we discussed in the Introduction, The entropy method requires a path lemma that roughly states that we can move a particle from one site to another paying a diffusive cost. This is not true for the Sierpinski gasket, due to the presence of {\em hot spots}: points that connects two huge parts of the graph that can not be avoided in order to move a particle from one of these parts to the other. For example, if we want to transport a particle from a site in $\varphi_0(K)$ to other site in $\varphi_1(K)$, the particle has to pass by point $\varphi_0(a_1)$, or by points $\varphi_0(a_2)$, $\varphi_1(a_2)$. 

The second alternative requires smoothness of the solutions of the hydrodynamic equation. Of course, since we do not have a differentiable structure in $K$, we do not expect the solutions of the hydrodynamic equation to be smooth. 

The $H_{-1}$ method is based on the heuristic argument leading to uniqueness of the hydrodynamic equation (\ref{echid}). 
Remember that the idea was to prove that, for two solutions $u_t$, $v_t$ of (\ref{echid}), $||u_t-v_t||_{-1}^2$ is decreasing in time.  The main point is that this inequality also holds {\em at the microscopic level}, that is, for two different versions $\xi_t^1$, $\xi_t^2$ of the zero-range process, or even between $u_t$ and $\xi_t$. 

Our task will be to put this formal arguing into a rigorous proof. Before doing that, we need some definitions.
We recall the formula for the norm in $H_{-1}(K)$ in terms of the Green function $\mc G$: for a function (or even a measure) $u: K \to \bb R$,
\[
||u||_{-1}^2 = \iint_{K \times K} u(x) u(y) \mc G(x,y) \nu(dx)\nu(dy).
\]

More important for us will be the discrete version of this formula: for $u: V_n \to \bb R$ such that $u(a_i)=0$,
\[
||u||_{-1,n}^2 = \frac{1}{3^{2n}} \sum_{x,y \in V_n} u(x) u(y) \mc G(x,y).
\]

Now we define what we understand by ``convergence in $H_{-1}$''. Let $\{\nu^n\}_n$ be a sequence of measures in $\Omega_n$. Let $u: K \to [0,\infty)$ be a given function. We say that $\nu^n$ converges to $u$ in the $H_{-1}$ sense if 
\[
\lim_{n \to \infty} \int ||\xi -u||_{-1,n}^2 \nu^n(d\xi) =0.
\]

A simple computation shows that the {\em local equilibrium} measures $\nu^n_{u(\cdot)}$ defined as the product measures in $\Omega_n$ with marginals
\[
\nu^n_{u(\cdot)} \big\{ \xi; \xi(x) =k\big\} = \nu_{u(x)}\big\{\xi;\xi(x)=k\big\}
\]
converge to $u$ in the $H_{-1}$ sense.	

For two given measures $\nu$, $\nu'$ in $\Omega_n^{zr}$, we define the relative entropy of $\nu$ with respect to $\nu'$ by
\[
H(\nu|\nu') =
\begin{cases}
\int \frac{d\nu}{d\nu'}\log \frac{d\nu}{d\nu'} d\nu', &\text{ if } \nu << \nu'\\
+\infty, &\text{ otherwise.}
\end{cases}
\] 

We also say that $\nu$ is {\em stochastically dominated} by $\nu'$ if there is a measure $\lambda$ in $\Omega_n\times \Omega_n$ such that 
\begin{itemize}
\item[i)] $\lambda(\xi,\Omega_n) = \nu(\xi)$
\item[ii)] $\lambda(\Omega_n,\xi) = \nu'(\xi)$
\item[iii)] $\lambda\{(\xi,\xi'); \xi(x) \leq \xi'(x) \text{ for any } x \in V_n^0\}=1$. 
\end{itemize}

Now we are ready to state the main result of this article.

\begin{theorem}
\label{t1}
Let $\{\nu^n\}_n$ be a sequence of probability measures in $\Omega_n^{zr}$, converging in the $H_{-1}$ sense to some bounded function $u_0: K \to [0,\infty)$. Assume that the hydrodynamic equation has a unique solution. Assume also the technical conditions:
\begin{itemize}
\item[i)] Under {\bf (SG)}, there are positive constants $\kappa$, $\rho$ such that $H(\nu^n|\nu_\rho) \leq \kappa 3^n$.
\item[ii)] Under {\bf (C)}, there are two positive constants $\rho<\rho'$ such that $\nu_\rho$ is stochastically dominated by $\nu^n$ and $\nu^n$ is stochastically dominated by $\nu_{\rho'}$ for any $n >0$.
\end{itemize}

Then, for any $t>0$, the distributions $\{\nu^n(t)\}_n$ at time $t>0$ of the rescaled process $\xi_t^n = \xi_{5^nt}$ in $\Omega_n$ starting from $\nu^n$, converge in the $H_{-1}$ sense to $u(t,\cdot)$, solution of the hydrodynamic equation (\ref{echid}).
\end{theorem}

Under Condition ii), we have $H(\nu^n|\nu_\rho) \leq H(\nu^{\rho'}|\nu_\rho)$ and therefore there is a constant $\kappa$ such that $H(\nu^n|\nu_\rho) \leq \kappa 3^n$ for any $n$.

\subsection{The $H_{-1}$-norm method: martingale representation}
In this section we will obtain a martingale representation for the $H_{-1}$-norm of $\xi_t^n$. Notice that, due to the time-scaling, the generator of $\xi_t^n$ is equal to $5^n L_{zr}$. For any function $F: [0,T] \times \Omega_n^{zr} \to  \bb R$, differentiable in time and linearly growing in $\xi$, Dynkin's formula states that
\[
M_t^{n,F} =: F(t,\xi_t^n) -F(0,\xi_0^n) - \int_0^t \big\{\partial_t+5^n L_{zr}\big\} F(s,\xi_s^n) ds
\]
is a martingale. A long and tedious, but totally elementary computation shows that, in fact,
\[
\big\{\partial_t + 5^n L_{zr}\big\} ||\xi_s-u_s^n||_{-1,n}^2 
	= -\frac{2}{3^n} \sum_{x \in V_n} \mc F\big(\xi_s^n(x), u_s^n(x)\big) 
	+ \frac{1}{3^{2n}} \sum_{x \in V_n} g\big(\xi_s(x)\big) \Delta_n \mc G(x),
\]
where $\mc G(x) = \mc G(x,x)$ and 
\[
\mc F(\xi,u) = \big(g(\xi)-\phi(u)\big)\big(\xi-u\big) - g(\xi).
\]

Observe that $\int \mc F\big(\xi(x),u\big) \nu_\rho(d\xi) \geq 0$ for any $\rho \in [0,\rho^*)$ and any $u \geq 0$. In particular, $\int \mc F\big(\xi(x),u\big) \nu_{u(\cdot)}^n(d\xi) \geq 0$ for any $x \in V_n$. Therefore, 
\begin{align*}
M_t^n 
	&= ||\xi_t-u_t^n||_{-1,n}^2-||\xi_0-u_0^n||_{-1,n}^2 \\
	&+\int_0^t \Big\{\frac{2}{3^n} \sum_{x \in V_n} \mc F\big(\xi_s^n(x), u_s^n(x)\big) 
	- \frac{1}{3^{2n}} \sum_{x \in V_n} g\big(\xi_s(x)\big) \Delta_n \mc G(x)\Big\}ds
\end{align*}
is a martingale. Since $M_0^n=0$, we have that $\bb E_n[M_t^n]=0$ for any $t>0$. Here and below, $\bb P_n$ denotes the distribution of the process $\xi_t^n$ starting from $\nu^n$, and $\bb E_n$ denotes expectation with respect to $\bb P_n$.

Taking the expectation with respect to $\bb P_n$ of the previous identity, we can obtain an expression for $\bb E_n||\xi_t - u_t^n||_{-1,n}^2$:
\begin{align*}
\bb E_n||\xi_t-u_t^n||_{-1,n}^2
	&= \bb E_n ||\xi_0 -u_0^n||_{-1,n}^2 
	- \bb E_n \int_0^t \frac{2}{3^n} \sum_{x \in V_n} \mc F\big(\xi_s^n(x), u_s^n(x)\big) ds\\
	&+ \bb E_n \int_0^t \frac{1}{3^{2n}} \sum_{x \in V_n} g\big(\xi_s(x)\big) \Delta_n \mc G(x)ds.\\
\end{align*}

Theorem \ref{t1} is an immediate consequence of the following two lemmas:

\begin{lemma}
\label{l1}
\[
\lim_{n \to \infty} \bb E_n \int_0^t \frac{2}{3^n} \sum_{x \in V_n} \mc F\big(\xi_s^n(x), u_s^n(x)\big) ds \geq 0.
\]
\end{lemma}

\begin{lemma}
\label{l2}
\[
\lim_{n \to \infty} \bb E_n \int_0^t \frac{1}{3^{2n}} \sum_{x \in V_n} g\big(\xi_s(x)\big) \Delta_n \mc G(x)ds \leq 0.
\]
\end{lemma}

In fact, from these two lemmas, we conclude that 
\[
\limsup_n \bb E_n||\xi_t-u_t^n||_{-1,n}^2 \leq \limsup_{n \to \infty} \bb E_n||\xi_0-u_0^n||_{-1,n}^2,
\]
and convergence in the $H_{-1}$ sense follows at once.  The argument behind the proof of these two lemmas is as follows. We will see that some sort of weak conservation of local equilibrium will allow us to replace in the previous expressions the functions $g(\xi_s(x))$ by $\phi(\xi_s^k(x))$, where the symbol $\xi_s^k(x)$ denotes the average of $\xi_s(y)$ over a small triangle containing $x$, paying a price that vanishes when $n \to \infty$ and then $k \to \infty$. In the same way we can substitute $\mc F(\xi_s(x),u_t^n(x))$ by $\phi(\xi_s^k(x)) - \phi(u_t^n(x))(\xi_s^k(x) -u_t^n(x))$. This last term is always positive, and we are in position to prove Lemma \ref{l1}. The proof of Lemma \ref{l2} is more subtle. Notice the huge factor $1/3^n$ in front of the average in Lemma \ref{l2}. Since $\mc G(x,y)$ satisfies $\sup_y||\mc G(\cdot,y)||_{1}^2 <+\infty$, we could guess that $\mc G(\cdot)$ is in $H_1(K)$. In that case, we should have $\Delta_n \mc G(x)/3^n \to 0$ as $n \to \infty$ in some convenient sense. It turns out that this is {\em not} the case. The function $\mc G(x)$ is extremely irregular, and in fact it can be proved that for each $x \in V^*$ fixed, $\Delta_n \mc G(x)/3^n \to 3/7$ as $n \to \infty$. Replacing $g(\xi_s(x))$ by $\phi(\xi_s^k(x))$ we will be able to perform an integration by parts in a small triangle of size $k$, therefore gaining a factor $3^k$ that will save the day at the end.  

We will devote the following two sections to the proof of each one of these two lemmas.

\section{The one-block estimate}
\label{s4}
The replacement mentioned in the previous section is known in the literature of interacting particle systems as the {\em one-block estimate}. Before stating the one-block estimate in a precise way, we need some definitions. 

Fix two integers $n \geq l \geq 0$. For $x \in V_n \setminus V_{n-l}$, we define $\mc T_n^l(x)$ as the set of points in $V_n$ contained in the triangle $\Delta(x_0,x_1,x_2)$, which contains $x$ and has vertices in $V_{n-l}$. In an equivalent way, $\mc T_n^l(x) = \mc T \cap V_n$, where $\mc T$ is the unique triangle of the form $\varphi_{i_{n-l}} \circ \cdots \circ \varphi_{i_1}(K)$ containing $x$. For $x \in V_{n-l}$, there are two possible choices for $\mc T_n^l(x)$, given by two triangles intersecting exactly at $x$. Rotating the graph in such a way that both triangles lie on the upper semiplane, we choose $\mc T_n^l(x)$ as the triangle at the right of $x$. The exact choice in this case is not important, the point is to choose each triangle the same number of times. 

\begin{theorem}[Local one-block]
Let us define the average number of particles $\xi_s^k(x)$ by
\[
\xi_s^k(x) = \frac{1}{|V_k|} \sum_{y \in \mc T_n^k(x)} \xi_s(y),
\]
where $|V_k|$ denotes the cardinality of $V_k$ (and also of $\mc T_n^k(x)$). Then,
\begin{equation}
\label{ec8}
\lim_{k \to \infty} \limsup_{n \to \infty} \sup_{x \in V_n} \bb E_n \Big| \int_0^t \Big\{ g(\xi_s(x)) - \phi\big(\xi_s^k(x)\big)\Big\}ds\Big| =0,  
\end{equation}

\begin{equation}
\label{ec9}
\lim_{k \to \infty} \limsup_{n \to \infty} \sup_{x \in V_n} \bb E_n \Big| \int_0^t \Big\{ \xi_s(x) g(\xi_s(x)) - \phi\big(\xi_s^k(x)\big)\big(1+ \xi_s^k(x)\big)\Big\}ds\Big| =0.
\end{equation}
\label{t2}
\end{theorem}

The word ``local'' comes from the fact that in the usual version of the one-block estimate, the arguments in the integrals are averaging against a smooth test function. This local version of the one-block estimate was introduced in \cite{JLS}. As stated in \cite{JLS}, this local one-block estimate is available only in dimension $d <2$. The Sierpinski gasket $K$ has Hausdorff dimension $d_H = \log(3/2) <2$. But this is not really the point. The local one-block estimate holds each time the scaling of the process is faster than the scaling of the number of points. In our present situation, the process scales like $5^n$ and the number of points scales like $3^n$, so the local one-block estimate will hold. 

\subsection{Proof of the one-block estimate}

We will take the proof of Theorem \ref{t2} from \cite{JLS}. In that paper the case on which condition {\bf (SG)} is satisfied is treated in detail, so here we focus on condition {\bf (C)}. Our first step is to introduce a cut-off function that prevents us to have too many particles at site $x$. Take $a > \rho'$. Then,
\begin{align*}
\bb E_n \Big| \int_0^t \Big\{ g(\xi_s(x)) - \phi\big(\xi_s^k(x)\big)\Big\}
	&{\mathbf 1}(\xi_s^k(x) \geq a) ds\Big| \leq \\
	&\leq \bb E_n \int_0^t \Big\{ g(\xi_s(x)) + \phi\big(\xi_s^k(x)\big)\Big\}{\mathbf 1}(\xi_s^k(x) \geq a) ds\\
	&\leq \bb E_{\rho'} \int_0^t \Big\{ g(\xi_s(x)) + \phi\big(\xi_s^k(x)\big)\Big\}{\mathbf 1}(\xi_s^k(x) \geq a) ds\\
	&\leq t \int \big\{ g(\xi(x)) + \phi\big(\xi^k(x)\big)\big\}{\mathbf 1}(\xi^k(x) \geq a) \nu_{\rho'}(d\xi)\\
	&\leq t \int \big\{g(\xi(x))^2 +\phi\big(\xi^k(x)\big)^2\big\}\nu_{\rho'}(d\xi) \nu_{\rho'}(\xi^k(x) \geq a).
\end{align*} 

The expectation in the last line is bounded in $k$. Moreover, by tha law of large numbers, $\xi^k(x)$ converges to $\rho'$ in probability as $k \to \infty$, and the probability in the last line goes to 0 as $k \to \infty$. Notice that this convergence is uniform in $x$. Therefore, we can introduce the indicatior function $\mathbf 1(\xi_s^k(x)\leq a)$ in (\ref{ec8}).

By assumption, the {\em entropy density} $H(\nu^n|\nu_\rho)/3^n$ is uniformly bounded in $n$, by a constant $\kappa<+\infty$. A simple computation shows that the same is true for $H(\nu^n|\nu_{h})$, the entropy with respect to the invariant measure of the process. It is well known that the entropy of a jump process with respect to the invariant measure is decreasing in time. Fix some reference time $T >0$. Denote by $\bb P_{inv}$ the law of the process $\xi_t$ up to time $T$, speeded up by $5^n$ and starting from the {\em invariant measure} $\nu_h$. Then, there is another constant $\bar \kappa$ depending only on $\kappa$ and $T$, such that $H(\bb P_n| \bb P_{inv})/3^n \leq \bar \kappa$ for any $n > 0$. To simplify the notation, let us define $\mc V_k(\xi,x)$ by
\[
\mc V_k(\xi,x) = \big\{g\big(\xi(x)\big) - \phi\big(\xi^k(x)\big)\big\}
	\mathbf 1(\xi_s^k(x) \leq a).
\]

We have ommited in the notation the dependence of $\mc V_k$ in $n$ and $a$.
By the entropy inequality,
\begin{align*}
\bb E_n\Big| 
	&\int_0^T
	\mc V(\xi_s,x) ds \Big| \leq \\
	&\leq \frac{\bar \kappa}{\gamma} + \frac{1}{\gamma 3^n} \log \bb E_{inv} \Big[\exp\Big\{\gamma 3^n \Big|\int_0^T \mc V_k(\xi_s,x) ds \Big| \Big\} \Big]. \\
\end{align*}

Using the elementary inequality $e^{|x|} \leq e^x+e^{-x}$, we can get rid of the modulus in the previous expression. Therefore, the limit in (\ref{ec8}) will be obtained if we prove that
\[
\lim_{k \to \infty} \limsup_{n \to \infty} \frac{1}{\gamma 3^n} \log \bb E_{inv} \Big[\exp\Big\{\pm\gamma 3^n \int_0^T \mc V_k(\xi_s,x) ds  \Big\} \Big] =0.
\]

By Feynman-Kac's formula, the logarithm of this expectation is bounded by the largest eigenvalue of the operator $5^n L \pm \gamma 3^n\mc V_k$, where the term $\mc V_k$ is understood as a multiplication operator. For simplicity, we will consider just the ``$+$'' sign in $\mc V_k$. By the variational formula for the largest eigenvalue of an operator in $\mc L^2(\nu_h)$, the previous expression is bounded by
\begin{equation}
\label{ec11}
T\sup_{f} \Big\{ \<\mc V_k,f\> - \frac{1}{\gamma}\Big(\frac{5}{3}\Big)^n \<\sqrt{f}, -L_{zr}\sqrt{f}\>\Big\}, 
\end{equation}
where the supremum is over all the densities $f$ with respect to $\nu_h$, and the inner product is with respect to $\nu_h$ as well. Notice first that $\mc V_k$ depends on $\xi(y)$ only through its values for $y \in \mc T_n^k(x)$. Taking the conditional expectation of $f$ with respect to $\mc F(\mc T_n^k(x))$, the $\sigma$-algebra generated by $\{\xi(y); y \in \mc T_n^k(x)\}$, and due to the fact that $\nu_h$ is a product measure, we can restrict the previous supremum to densities in the configuration space $\bb N_0^{\mc T_n^k(x)}$, which is homeomorphic to $\bb N_0^{V_k}$. By positivity of $\<\sqrt{f}, -L_{zr}\sqrt{f}\>$, we can also change $L_{zr}$ by the generator of a zero-range process restricted to the triangle $\mc T_n^k(x)$. We will denote this generator simply by $L$, since no risk of confusion will appear by the fact that $L$ depends on $n,k$ and $x$. 
In this way, we have reduced the initial problem into a problem on a finite graph, which in our case is equal to $V_k$. Moreover,
due to the presence of the indicator function $\mathbf 1(\xi^k(x) \leq a)$ in the definition of $\mc V_k$, we can restrict ourselves to a finite state space, namely $\{\xi \in \bb N_0^{V_k} ; \xi^k \leq a\}$. Notice that at this point, $\xi^k=\xi^k(x)$ does not depend on $x \in V_k$. Since now the supremum is over a compact set, we can exchange the supremum and the limit as $n \to \infty$ to obtain
\[
\limsup_{n \to \infty} \sup_{f} \Big\{ \<\mc V_k,f\> - \frac{1}{\gamma}\Big(\frac{5}{3}\Big)^n \<\sqrt{f}, -L\sqrt{f}\>\Big\}
	= \sup_{\<\sqrt{f}, -L \sqrt{f}\>=0} \<\mc V_k,f\>.
\]

Now we have to identify the densities for which $\<\sqrt{f}, -L \sqrt{f}\>=0$. A simple computation shows that, in this case, $f$ is constant over the sets 
\[
\Omega_{k,l} =\{\xi \in \bb N_0^{V_k}; \sum_{x \in V_k} \xi(x) =l\}.
\]

The restriction $\xi^k \leq a$ imposes $l \leq a |V_k|$. Let us define the measures $\nu_{k,l}$ by taking $\nu_{k,l}(\cdot) = \nu_\rho(\cdot|\xi^k=l|V_k|)$. Notice that these measures do not depend on the value of $\rho$, and they are also exchangeable. Then, the previous supremum is equal to 
\[
\sup_{l \leq a|V_k|} \int \mc V_k(\xi,x) \nu_{k,l}(d\xi).
\]

For $l \leq a|V_k|$, the indicator function $\mathbf 1(\xi^k \leq a)$ is identically equal to 1. Therefore, we are left with
\[
\sup_{l \leq a|V_k|} \int \big\{g(\xi(x)) - \phi(l/|V_k|)\big\}d\nu_{k,l}.
\]

But this last quantity goes to $0$ as $k \to \infty$ by the equivalence of ensembles, which states that the {\em grancanonical} measures $\nu_{k,l}$ approach the {\em canonical} measures $\nu_{l/|V_k|}$, uniformly in compact sets of the real line.

The case on which we take the ``$-$'' sign in front of $\mc V_k$ is totally analogous. In this way we have finished the proof of (\ref{ec8}).
The proof of this theorem when we take $\xi(x) g(\xi(x))$ instead of $g(\xi(x))$ is entirely analogous, and we left it to the interested reader.

\begin{remark}
The same estimate remains true if we consider functions of the form $g(\xi_t^n(x))F(t)$, where $F:[0,T] \to \bb R$ is bounded. In that case the limit is uniform over sets of the form $\{\sup_{t \in [0,T]} |F(t)| \leq K\}$. It is enough to replace the variational formula in (\ref{ec11}) by $\int_0^T \lambda_n(t)dt$, where $\lambda_n(t)$ is the largest eigenvalue of the operator $5^n L \pm \gamma 3^n \mc V_{k,t}$.
\end{remark}

\subsection{Proofs of Lemma \ref{l1} and Lemma \ref{l2}}

Now we are in position to prove Lemma \ref{l1}. By the one-block estimate, we can write
\begin{multline*}
\bb E_n \int_0^t \frac{2}{3^n} \sum_{x \in V_n} \mc F(\xi_s^n(x),u_s^n(x))ds = \\
	=\bb E_n \int_0^t \frac{2}{3^n}\sum_{x \in V_n} \big(\phi(\xi_s^k(x)) -\phi(u_s^n(x))\big)\big(\xi_s^k(x)-u_s^n(x)\big)ds
\end{multline*}
plus a rest that vanishes as $n \to \infty$ and then $k \to \infty$. But this last term is always positive, which proves Lemma \ref{l1}.

In order to prove Lemma \ref{l2}, we start proving that $\Delta_n \mc G(x) /3^n$ is uniformly bounded. Remember that $\mc G(x,x) \geq \mc G(x,y)$ for any $x,y \in K$. Since $\Delta_n \mc G(x,x) =-3^n$, we conclude that $\mc G(x,x) \leq \mc G(x,y) +3^n$ for $x$, $y$ in $V_n$ with $y \sim_n x$. Therefore,
\begin{align*}
\mc G(x,y) 
	&\leq \mc G(x,x) \leq \mc G(x,y),\\
\mc G(x,y) 
	&\leq \mc G(y,y) \leq \mc G(x,y),\\
\end{align*}
where we have obtained the second line by interchanging the roles of $x$ and $y$. We conclude that $|\mc G(x) -\mc G(y)| \leq 3^n$ for $x \sim_n y$, and $|\Delta_n \mc G(x)/3^n| \leq 4$. Now we are able to use the one-block estimate to rewrite the expectation in Lemma \ref{l2} as
\[
\bb E_n \int_0^t\frac{1}{3^{2n}} \sum_{x \in V_n} g\big(\xi_s(x)\big)\Delta_n \mc G(x) ds=
	\bb E_n \int_0^t\frac{1}{3^{n}} \sum_{x \in V_n} \phi\big(\xi_s^k(x)\big)\Delta_n \mc G(x)/3^n ds	
\]
plus a rest that vanishes as $n \to \infty$ and then $k \to \infty$.
Notice that the function $\phi(\xi_s^k(x))$ is constant in $\mc T_n^k(x)$. Therefore, we can integrate by parts (in this discrete context just a summation) the function $\Delta_n \mc G(x)$ in $\mc T_n^k(x)$ to obtain that the previous expression is equal to
\[
\bb E_n \int_0^t\frac{1}{3^{n}} \frac{1}{|V_k|}\sum_{x \in V_n} \phi\big(\xi_s^k(x)\big)\sum_{i=0,1,2} \partial_{n,k}^i \mc G(x)ds,
\]
where the symbols $\partial_{n,k}^i G(x)$ denote the {\em outer} normal derivative of $\mc G(x)$ computed on the three vertices $a_i^{n,k}(x)$ of the triangle $\mc T_n^k(x)$, defined by
\[
\partial_{n,k}^i G(x) = (5/3)^n \sum_{\substack{y \sim_n a_i^{n,k}(x)\\y \notin \mc T_n^k(x)}} \big\{\mc G(y)-\mc G(a_i^{n,k}(x))\big\}.
\]

Notice that the ordering of the three vertices $a_i^{n,k}(x)$ is not relevant here. We have also gained a factor $|V_k|^{-1}$ in this integral. Now we just need to prove that the normal derivatives of $\mc G(x)$ defined in this way are uniformly bounded. But the same arguments used to bound $\Delta_n \mc G(x)$ can be repeated here, to get a bound of the form $|\partial^i_{n,k} \mc G(x)| \leq 2$ for any $x \in V_{n-k}$ and any $k \leq n$.

\appendix

\section{The behavior of the Green function at the diagonal}
\label{s5}

\begin{figure}
\centering
\vspace{-100pt}
\begin{picture}(250,250)
	\includegraphics{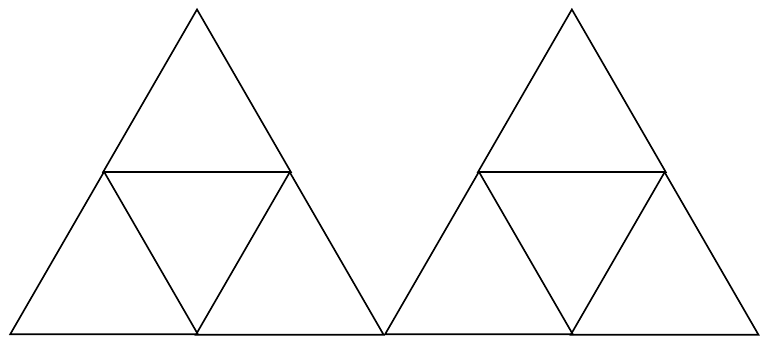}
	\put(-390,0){$x_1$}
	\put(-445,0){$y_2$}
	\put(-500,0){$x_0$}
	\put(-555,0){$y_2'$}
	\put(-610,0){$x_1'$}
	\put(-590,60){$y_0'$}
	\put(-515,60){$y_1'$}
	\put(-480,60){$y_1$}
	\put(-405,60){$y_0$}
	\put(-555,115){$x_2'$}
	\put(-445,115){$x_2$}
	\put(-545,-15){Figure A: A neighborhood of $x_0$}
\end{picture}
\vspace{10pt}	
	\label{fig1}
\end{figure}

In this Appendix we study the behavior of the function $\mc G(x)$. We want to compute $\Delta_n \mc G(x) /3^n$ for $x \in V_n$ and we want to obtain its asymptotic behavior. The idea is to obtain an iterative formula for $\Delta_{n+1} \mc G(x)$ in terms of the values of $\mc G(x,y)$ for sites $y$ in the neighborhood of $x$. 
In Fig. 1 we have taken a point $x_0 \in V_n$ and we have named $x_i$, $x_i'$, $i=1,2$ each of the four neighbors of $x_0$ in $V_n$. We have also drawn points $y_i$, $y_i'$, $i=0,1,2$, which correspond to the points in $V_{n+1}$ belonging to the two triangles of side $2^{-n}$ in $V_n$, meeting at $x_0$. We claim that $\mc G(y_i,x_j)$ can be computed in terms of the six numbers $G^n_{ij} = \mc G(x_i,x_j)$ (there are 9 combinations for $i$, $j$, but remeber that $G^n_{ij}=G^n_{ji}$). In fact, let us construct $\mc G(y,y_0)$ for $y \in V_{n+1}$, for example. Defining
\[
\mc G_0(y,y_0) =
\begin{cases}
3/10, &y=y_0\\
1/10, &y=y_1, y_2\\
0, &\text{otherwise, }
\end{cases}
\]
we see that $\Delta_{n+1} (3/5)^{n+1} \mc G_0(y,y_0) = -3^n \delta(y,y_0)$, except for $y =x_i$, $i=0,1,2$. Therefore, the true Green function $\mc G(y,y_0)$ is a linear combination between $\mc G_0(y,y_0)$ and $\mc G(y,x_i)$, $i=0,1,2$. In fact, $\Delta_{n+1} (3/5)^{n+1} \mc G_0(x_i,y_0)= 3^{n+1}h^0(x_i)$, where
\[
h^0(x_i)=
\begin{cases}
1/5, &i = 0\\
2/5, &i=1,2.
\end{cases}
\]

Therefore, we have the following formula for $\mc G(y,y_0)$:
\[
\mc G(y,y_0) = \Big(\frac{3}{5}\Big)^{n+1} \mc G_0(y,y_0) + \sum_{i=0,1,2} h^0(x_i) \mc G(y,x_i).
\]

Notice, as well, that $\mc G(y_i,x_j)$ can be computed from $G_{ij}^n$ by harmonic continuation:
\[
\mc G(y_i,x_j) = \frac{1}{5}\Big\{ 2\sum_{k=0,1,2} G_{kj}^n -G_{ij}^n\Big\}.
\]

Combining these two formulas, we can obtain the values of $\mc G(y_i,y_j)$ in terms of the numbers $\{G_{ij}^n\}$ for any $i,j$. We give the formulas for $\mc G(y_0,y_0)$, $\mc G(y_1,y_0)$; the other formulas can be obtained by cyclic permutations of $\{0,1,2\}$:
\[
\mc G(y_0,y_0) = \frac{3}{10} \Big(\frac{3}{5}\Big)^{n+1} + \frac{1}{25} \big\{ G_{00}^n + 4G_{11}^n +4 G_{22}^n + 4G_{01}^n +8G_{12}^n +4G_{20}^n\big\}, 
\]
\[
\mc G(y_0,y_1) = \frac{1}{10} \Big(\frac{3}{5}\Big)^{n+1} + \frac{1}{25} \big\{ 2G_{00}^n + 2G_{11}^n +4 G_{22}^n + 5G_{01}^n +6G_{12}^n +6G_{20}^n\big\}.
\]

Of course, similar formulas hold for the {\em left-side points} in Figure \ref{fig1}. As an application of these formulas, we can try to compute $\Delta_{n+1} \mc G(x_0)$ in terms of $\Delta_n \mc G(x_0)$:

\begin{multline}
\label{ec10}
5^{n+1}\Big(\mc G(y_1,y_1)+\mc G(y_2,y_2) - G_{00}^n\Big) 
	= 3^{n+1} \cdot\frac{3}{5}+\\
	  +5^{n-1}\Big\{ 5(G_{11}^n+G_{22}^n-2G_{00}^n)+12(G_{01}^n+G_{02}^n-G_{00}^n) + 8(G_{12}^n-G_{00})\Big\}.
\end{multline}

Adding the symmetric term coming from the left-hand side of Figure \ref{fig1}, we see that
\begin{multline*}
\Delta_{n+1} \mc G(x_0) = 3^{n+1}\cdot \frac{6}{5} + \Delta_n \mc G(x_0) + \frac{12}{5} \Delta_n \mc G(x_0,x_0)\\
	 + 5^n \cdot \frac{8}{5} \big\{\mc G(x_1,x_2)+\mc G(x_1',x_2')-2\mc G(x_0,x_0)\big\}.
\end{multline*}

In this expression, there is a new term appearing. Let us define 
\[
\Gamma_n(x_0) = 5^n \big\{\mc G(x_1,x_2)+\mc G(x_1',x_2')-2\mc G(x_0,x_0)\big\}.
\]

Since $\Delta_n \mc G(x_0,x_0) = -3^n$, we have the formula
\[
\Delta_{n+1} \mc G(x_0) = 3^{n+1}\cdot \frac{2}{5} + \Delta_n \mc G(x_0) + \frac{8}{5}\Gamma_n(x_0).
\]

Now let us compute $\Gamma_{n+1}(x_0)$:
\[
\Gamma_{n+1}(x_0) = -3^{n+1}\cdot \frac{1}{5} + \frac{2}{5}\Delta_n \mc G(x_0) + \Gamma_n(x_0).
\]

This establish a linear recurrence formula for the pair $\{\Delta_n \mc G(x_0), \Gamma_n(x_0)\}$. Due to the factor $3^{n+1}$, it is natural to define $a_n = \Delta_n \mc G(x_0)/3^n$, $b_n = \Gamma_n(x_0)/3^n$. For $a_n$, $b_n$, the recurrence formula reads
\begin{align*}
a_{n+1} &= 2/5 + a_n/3 + 8b_n/15\\
b_{n+1} &= -1/5+ 2a_n/15 +b_n/3.\\
\end{align*}

This recursion formula can be written in vectorial terms as $\mathbf a_{n+1} = \mathbf w + M \mathbf a_n$, with \[
M= \left(
\begin{array}{cc}
1/3 &8/15\\
2/15 &1/3\\
\end{array}
\right).
\]

The eigenvalues of this matrix are $\lambda_1=3/5$, $\lambda_2=1/15$. In particular, for any initial value of $a_n$, $b_n$ (remember that $x_0 \in V_n$, so the sequence does not start at $n=1$) we have convergence to a unique fixed point, given by $\mathbf a = (I-M)^{-1} \mathbf w$. In our case, $\mathbf w = (2/5, -1/5)$, and $\mathbf a=(3/7,-3/14)$. In particular,
\[
\lim_{n \to \infty} \Delta_n \mc G(x_0)/3^n = 3/7.
\]

Notice that the limit is {\em positive} and {\em does not depend} on $x_0$. Remember that $\Delta_n \mc G(x_0,x_0)/3^n = -1$. This remarkable fact shows the high irregularity of the function $\mc G(x)$. In contrast to it, for the unit interval $[0,1]$, $\mc G(x) = x(1-x)$, so $\mc G(x)$ is smooth and concave.

In order to obtain a recursive formula for the partial derivatives of $\mc G(x)$, we focus our attention on the right side of Figure \ref{fig1}. Let us take a look at equation (\ref{ec10}).  We see that most of the work has already been done. In fact, defining 
\begin{align*}
\alpha_n 
	&= (5/3)^n(G_{11}^n+G_{22}^n-2G_{00}^n)\\
\beta_n
	&= (5/3)^n(G_{12}^n-G_{00}^n)\\
\gamma_n
	&= (5/3)^n(G_{01}^n+G_{20}^n-2G_{00}^n),
\end{align*}
we obtain the following recursion formulas:
\begin{align*}
\alpha_{n+1}
	&=3/5+1/3 \alpha_n + 8/15 \beta_n + 4/5 \gamma_n\\
\beta_{n+1}
	&=1/10+ 2/15 \alpha_n +1/3 \beta_n + 2/5 \gamma_n\\
\gamma_{n+1} 
	&= \gamma_n.
\end{align*}

In particular, we obtain the same linear recursion as before, but now with a different vector $\mathbf w = (3/5 +4/5 \gamma,1/10+2/5 \gamma)$. Again, we have convergence of $(\alpha_n,\beta_n)$ to the solution of $\mathbf a = \mathbf w +M \mathbf a$. In our case, $\lim_n \alpha_n =\alpha$, with $\alpha = 17/14 +2\gamma$, where $\gamma=\gamma_n$ for some suitable $n$. Notice that we recover our previous computation $\Delta_n \mc G(x)/3^n \to 3/7$ by noticing that the Laplacian at $x$ is the sum of the two partial derivatives at $x$, and that the corresponding $\gamma$'s add up to $-1$ in that case. 

We see that the behavior of $\Delta_n \mc G(x)$ is the {\em worst possible}, in the sense that in one hand we have the trivial bound $\Delta_n \mc G(x)/3^n \leq c_1$ for any $n$ and $c_1=4$, and in the other hand we have the lower bound $\Delta_n \mc G(x)/3^n \geq c_2$ for any $c_2 \leq 3/7$ and any $n$ large enough.

\section*{Acknowledgements}
M.J. was supported by the Belgian Interuniversity Attraction Poles Program P6/02, through the network NOSY (Nonlinear systems, stochastic processes and statistical mechanics).

\end{document}